\documentclass[preprint,nofootinbib,amsmath,amssymb,amsfonts,aps,prd]{revtex4-1}

\usepackage{graphicx}
\usepackage{subfigure}
\usepackage[bookmarks=true, pdfstartview={FitH}, bookmarksopen=true, bookmarksopenlevel=1, linktoc=page]{hyperref}
\usepackage{amsthm}
\usepackage{mathrsfs}

\newcommand{\be}{\begin{equation}}
\newcommand{\ee}{\end{equation}}
\newcommand{\ba}{\begin{aligned}}
\newcommand{\ea}{\end{aligned}}

\usepackage{amsthm}
\newtheorem*{theorem}{Theorem}
\newtheorem*{lemma}{Lemma}
\newtheorem{condition}{Condition}

\begin{document}

\pagestyle{myheadings}

\title{Turning Point Instabilities for Relativistic Stars and Black Holes}
\author{Joshua S. Schiffrin}
\email{schiffrin@uchicago.edu}
\author{Robert M. Wald}
\email{rmwa@uchicago.edu}
\affiliation{Enrico Fermi Institute and Department of Physics \\
  The University of Chicago \\
  5640 S. Ellis Ave., Chicago, IL 60637, U.S.A.}
\date{\today}

\begin{abstract} 

In the light of recent results relating dynamic and thermodynamic stability of relativistic stars and black holes, we re-examine the relationship between ``turning points''---i.e., extrema of thermodynamic variables along a one-parameter family of solutions---and instabilities. We give a proof of Sorkin's general result---showing the existence of a thermodynamic instability on one side of a turning point---that does not rely on heuristic arguments involving infinite dimensional manifold structure. We use the turning point results to prove the existence of a dynamic instability of black rings in $5$ spacetime dimensions in the region where $c_J > 0$, in agreement with a result of Figueras, Murata, and Reall.

\end{abstract}

\maketitle

\section{Introduction} \label{intro}

As defined more precisely below (see section \ref{proof}), a ``turning point'' along a one-parameter family of thermodynamic equilibrium configurations is a point at which the first derivatives along the family of the fundamental thermodynamic variables vanish. Turning points have been widely used to find instabilities of stars in general relativity; see \cite{Friedmanbook,FIS} and references cited therein. For example, it has long been known that for a sequence of static spherically symmetric stars with equation of state of the form $P=P(\rho)$, an instability sets in at a point of maximum mass in the sequence. Similar criteria have been applied to determine the presence\footnote{\footnotesize It also was argued in \cite{FIS} that instabilities should not occur prior to turning points, but numerical studies \cite{Takami} have indicated that the onset of instability does occur prior to the turning point. \medskip} of instabilities of rotating stars \cite{FIS}. Turning point methods have also been used in the study of instabilities of black holes and higher dimensional black objects in vacuum general relativity \cite{Reall,AL}.

More than 30 years ago, Sorkin \cite{Sorkin} gave a very general argument that the existence of a turning point along a sequence of thermodynamic equilibrium states implies the presence of a thermodynamic instability on one side of the turning point. However, Sorkin's argument has some heuristic elements involving the assumption of an infinite-dimensional manifold structure on the space of solutions, and he also makes other general assumptions that would need to be verified in particular applications. One purpose of this note is to give a simple proof of Sorkin's result that does not require infinite dimensional manifold structure or other heuristic arguments. Another purpose is to determine cases for relativistic stars and black holes where the presence of a turning point implies a dynamic instability, in light of the general analyses of \cite{HollandsWald} and \cite{GSW} relating dynamic and thermodynamic instabilities. A final purpose is to relate the general turning point stability criterion to the criterion obtained in \cite{Reall} for the instability of black rings in 5 dimensions, namely, positivity of the heat capacity at fixed angular momenta, $c_J$.

We now describe the general framework in which we are working, and explain how relativistic perfect fluid stars and black holes fit within this framework. 

We consider systems that are described by a set of local fields, $\phi$, on spacetime satisfying field equations with a well-posed initial value formulation. We assume that we have notions of the total mass/energy, $M[\phi]$, and the total entropy, $S[\phi]$, of any solution $\phi$. We assume further that there are a finite number, $p$, of other ``fundamental conserved quantities,'' $\left(X^1[\phi], \dots, X^p[\phi]\right)$, such that a first law of thermodynamics holds for all perturbations of a solution in thermodynamic equilibrium. Here, a solution, $\phi$, that is in dynamic equilibrium is said to be in \emph{thermodynamic equilibrium} if and only if $\delta S=0$ for all first order perturbations, $\delta\phi$, that have $\delta M=\delta X^i=0$, whereas a \emph{first law of thermodynamics} is said to hold if there exist constants $\mathcal T$ and $Y_i$ (which depend upon the thermodynamic equilibrium solution $\phi$) such that for all perturbations, we have
\be \label{firstlaw}
\delta M = \mathcal T \delta S + Y_i \delta X^i \,,
\ee
where a sum over $i$ is understood. Note that if $\mathcal T \neq 0$, the satisfaction of the first law \eqref{firstlaw} for perturbations of a solution $\phi$ implies that $\phi$ is a thermodynamic equilibrium solution, since clearly $\delta S = 0$ whenever $\delta M=\delta X^i=0$.

Two important examples of such systems are relativistic, perfect fluid stars and black holes. As discussed in detail in \cite{GSW}, for relativistic, perfect fluid stars, 
$\phi$ consists of the metric $g_{ab}$, the fluid velocity $u^a$ (satisfying $u^a u_a = -1$), the particle number density $n$, and the entropy per particle $s$. The energy density of the fluid is then determined by a specified equation of state $\rho=\rho(n,s)$, and the pressure is given by the ``Gibbs-Duhem'' relation
\be
P = -\rho+\mu n + T s n \, ,
\ee
where
\be
T \equiv \frac{1}{n} \frac{\partial\rho}{\partial s} \, , \quad \quad \quad \mu \equiv \frac{\partial\rho}{\partial n} - Ts \, 
\ee
are the local temperature and chemical potential of the fluid.
The total mass, $M$, is taken to be the ADM mass of the spacetime, and the total entropy, $S$, is given by
\be
S = \int_\Sigma s n u^a \epsilon_{abcd} \, ,
\ee
where $\Sigma$ is a Cauchy surface and $\epsilon_{abcd}$ is the spacetime volume element associated with 
$g_{ab}$. The additional ``fundamental conserved quantities'' are the ADM angular momentum, $J$, and the total particle number, $N$, defined by 
\be
N = \int_\Sigma  n u^a \epsilon_{abcd} \, .
\ee
A solution is said to be in dynamic equilibrium if it is stationary with killing vector $t^a$, is axisymmetric with killing vector $\varphi^a$, and has circular flow, i.e., $u^a$ is of the form 
\be
u^a = \frac{t^a+\Omega \varphi^a}{|v|}
\ee
for some function $\Omega$. It was proven in \cite{GSW} that a star in dynamic equilibrium will be in thermodynamic equilibrium if and only if $\Omega$, $\widetilde T \equiv |v| T$, and $\widetilde\mu \equiv |v| \mu$ are constant throughout the star. In other words, a star in dynamic equilibrium will be in thermodynamic equilibrium if and only if it rotates rigidly and has uniform redshifted temperature and redshifted chemical potential. For stars in thermodynamic equilibrium, the first law of thermodynamics holds in the form
\be
\delta M = \widetilde T \delta S + \Omega \delta J + \widetilde \mu \delta N \, .
\ee

For black holes in vacuum general relativity, $\phi$ is simply the spacetime metric $g_{ab}$. A solution is said to be in dynamical equilibrium if it is stationary with asymptotically timelike killing vector $t^a$. If the black hole is not static, then by the rigidity theorem \cite{HIW, MI}, it must also be axisymmetric with rotational killing vector $\varphi^a$. The total mass, $M$, is again taken to be the ADM mass, whereas the total entropy, $S$, is taken to be $A/4$, where $A$ is the surface area of the event horizon. An additional\footnote{\footnotesize For vacuum general relativity in spacetime dimensions $D>4$, there will be independent rotational planes and, hence, multiple angular momenta, $J^i$; for general relativity with suitable long range matter fields, there will also be conserved charges $Q^i$. \medskip} ``fundamental conserved quantity'' is the ADM angular momentum $J$. Stationary black holes satisfy the first law of black hole mechanics
\be
\delta M = \frac{\kappa}{8\pi} \delta A + \Omega_H \delta J \, ,
\ee
where $\kappa$ is the surface gravity of the black hole and $\Omega_H$ is the angular velocity of the horizon. Thus, for black holes, $\mathcal T = \kappa/2 \pi$. It follows from the first law that {\it all} stationary black holes in general relativity with $\kappa > 0$ are in thermodynamic equilibrium.

In the next section, we define our notion of ``turning points,'' and give a proof of Sorkin's result that does not rely on infinite dimensional manifold structure. We then give several applications.

\section{Turning Points and Thermodynamic Instability}\label{proof}

We now return to the general framework outlined in the previous section. Let $\phi(\lambda)$ be a
smooth 1-parameter family of solutions in thermodynamic equilibrium, with mass, $M(\lambda)$, and 
other conserved quantities $X^i (\lambda)$. Then $\phi(0)$ is said to be a \emph{turning point} of the family if at $\lambda = 0$ we have $d \phi/d \lambda \neq 0$ but
\be
\frac{d M}{d \lambda} = \frac{d X^i}{d \lambda} = 0 \, .
\label{tp}
\ee
By the first law, eq.~\eqref{firstlaw}, we also have $d S /d \lambda = 0$ at a turning point; more generally, a turning point may be defined by the simultaneous vanishing of any $(p+1)$ of the $(p+2)$ quantities $d S /d \lambda , d M /d \lambda , d X^1/d \lambda, \dots, d X^p/ d \lambda$. 

An alternative notion of turning points used by \cite{Reall, Read} may be given as follows. Suppose we have a smooth $(p+1)$-parameter family, $\phi(\lambda^0,\dots,\lambda^p)$ of thermodynamic equilibrium solutions that is nondegenerate in the sense that the perturbations $\partial \phi/\partial \lambda^A$ for $A=0,\dots,p$ are linearly independent. We say that this family has a {\it turning point} at $(\lambda^0,\dots,\lambda^p)$ if at that point the determinant of the Jacobian matrix, $\mathcal J^A_{\phantom AB}$, vanishes,
\be
\left|\mathcal J^A_{\phantom AB}\right| = \left|\frac{\partial X^A}{\partial \lambda^B}\right| = 0 \, ,
\label{jac}
\ee
where we have written $X^0 = M$, and $A$ and $B$ range from $0$ to $p$. The existence of a turning point in this sense can be seen to imply the existence of a turning point in the previous sense as follows: If the Jacobian determinant vanishes then $\mathcal J^A_{\phantom AB}$ has an eigenvector, $V^A$, with eigenvalue zero,
\be
\mathcal J^A_{\phantom AB} V^B=0 \, .
\ee
It follows immediately that any curve passing through the point $(\lambda^0,\dots,\lambda^p)$ with tangent vector $V^A$ will satisfy $d\phi/d\lambda \neq 0$ and $dX^A/d\lambda = 0$ at that point. Conversely, if a $(p+1)$-parameter family of thermodynamic equilibrium solutions contains a curve with a turning point in the sense of eq.~\eqref{tp}, then that point will be a turning point of the family in the sense of eq.~\eqref{jac}.

As stated above, by definition the condition for thermodynamic equilibrium is that $S$ be an extremum at fixed $M$ and $X^i$. Similarly, by definition, the condition for thermodynamic stability is that $S$ must be a local maximum at fixed $M$ and $X^i$. In particular, a thermodynamic equilibrium solution $\phi$ will be thermodynamically unstable if there exists a field variation to second order, i.e., $\phi \to \phi + \alpha \delta \phi + \alpha^2/2 \, \delta^2 \phi$, that keeps $M$ and $X^i$ fixed to both first and second order and for which $\delta^2 S > 0$. An alternative criterion that depends only on the first order variation $\delta \phi$, may be given as follows. The quantity
\be \label{Bdef0}
\mathcal B \equiv \delta^2 M - \mathcal T \delta^2 S - Y_i \delta^2 X^i
\ee
does not depend on the second order perturbation $\delta^2 \phi$, since, for fixed $\delta \phi$, the difference between two second order perturbations $\delta^2 \phi$ and $\delta^2 \phi'$ satisfies the linearized perturbation equations and hence must satisfy \eqref{firstlaw}. Thus, $\mathcal B$ depends only on the background equilibrium solution $\phi$ and the first order perturbation $\delta \phi$. 
Assuming $\mathcal T > 0$, we see immediately that any field variation that keeps $M$ and $X^i$ fixed to both first and second order and for which $\delta^2 S > 0$ will satisfy $\mathcal B < 0$. Conversely, suppose that the following condition holds for the theory under consideration: 
\begin{condition} \label{condA}
Let $\phi$ be a thermodynamic equilibrium solution and let $(c^0, \dots, c^p)$ be constants. Then there exists a linearized perturbation, $\delta \phi$, of $\phi$ such that $\delta M = c^0$ and $\delta X^i = c^i$ for $i=1,\dots,p$.
\end{condition}

\noindent
Then it is easy to see that if
$\mathcal B < 0$ for some first order perturbation $\delta \phi$ for which $\delta M = \delta X^i = 0$, we can adjust the second order perturbation by addition of a solution to the linearized equations so as to satisfy $\delta^2 M = \delta^2 X^i = 0$; hence, for this adjusted field variation, we have $\delta^2 S > 0$ with $M$ and $X^i$ fixed to both first and second order. Thus, for any theory that satisfies Condition \ref{condA}, if $\phi$ is a thermodynamic equilibrium solution with $\mathcal T > 0$ and if one can find a first order perturbation, $\delta \phi$, of $\phi$ for which $\delta M = \delta X^i = 0$ and $\mathcal B < 0$, then $\phi$ is thermodynamically unstable.

Condition \ref{condA} merely requires that one be able to find {\it some} perturbation of a thermodynamic equilibrium solution that produces any desired values of $\delta M $ and $\delta X^i$. It can be seen to hold for relativistic stars by virtue of the lemma of Appendix A of \cite{GSW}. We believe that it holds for black holes, although we have not attempted to prove this.

It is useful to view $\mathcal B$ in the following manner as a symmetric quadratic form on perturbations off of a thermodynamic equilibrium solution $\phi$. Let $\phi(\lambda_1, \lambda_2)$ be a two-parameter family of solutions with $\phi(0,0) = \phi$. Then the quadratic form $\mathcal B$ is given by
\be \label{Bdef}
\mathcal B\left[\phi; \delta_1 \phi, \delta_2 \phi\right] \equiv \delta_1 \delta_2 M - \mathcal T \delta_1 \delta_2 S - Y_i \delta_1\delta_2 X^i \, ,
\ee
where $\delta_1 \phi \equiv \partial \phi/\partial \lambda_1 |_{\lambda_1=\lambda_2=0}$, $\delta_1\delta_2 M \equiv \partial^2 M/\partial \lambda_1 \partial \lambda_2 |_{\lambda_1=\lambda_2=0}$, etc. It is important to note that $\mathcal B$ depends linearly on $\delta_1 \phi$ and $\delta_2 \phi$ and, by the first law, does not depend on $\partial^2 \phi/\partial \lambda_1 \partial \lambda_2$ or any other $\lambda_1$ or $\lambda_2$ derivatives of $\phi$ higher than first. Our above criterion for the thermodynamic instability of $\phi$ is the existence of a perturbation $\delta \phi$ for which $\delta M = \delta X^i = 0$ and
\be
\mathcal B\left[\phi; \delta \phi, \delta \phi\right] < 0 \, .
\ee

We now prove the following lemma concerning $\mathcal B$:

\begin{lemma}
Let $\phi(\lambda_1, \lambda_2)$ be a smooth 2-parameter family of solutions such that $\phi(\lambda_1, 0)$ is a thermodynamic equilibrium solution for all $\lambda_1$. Then, at $\lambda_2=0$, we have for all $\lambda_1$ 
\be
\mathcal B\left[\phi(\lambda_1,0); \delta_1 \phi, \delta_2 \phi \right] = \frac{d\mathcal T}{d\lambda_1} \delta_2 S + \frac{d Y_i}{d\lambda_1} \delta_2 X^i \, .
\label{Blemma}
\ee
\end{lemma}

\begin{proof}
By the first law of thermodynamics \eqref{firstlaw}, at $\lambda_2 = 0$ we have for all $\lambda_1$
\be
0 = \delta_2 M - \mathcal T(\lambda_1) \delta_2 S -  Y_i(\lambda_1) \delta_2 X^i \, .
\ee
Differentiating this equation with respect to $\lambda_1$ and using the definition, \eqref{Bdef}, of  $\mathcal B$ we immediately obtain \eqref{Blemma}.
\end{proof}

We now prove the turning point theorem \cite{Sorkin}:

\begin{theorem}
For a theory that satisfies Condition \ref{condA}, let $\phi(\lambda)$ be a smooth 1-parameter family of thermodynamic equilibrium solutions that has a turning point at $\lambda=0$ and is such that $\mathcal T > 0$ at $\phi(0)$. Suppose further that we have
\be
 \sigma \equiv 
\left.\frac{d}{d\lambda}\left(\frac{d{\mathcal T}}{d \lambda} \frac{d S}{d \lambda}+
 \frac{d{Y_i}}{d \lambda} \frac{d X^i}{d \lambda}\right) \right|_{\lambda=0} > 0 \, .
 \label{instab}
 \ee
 Then there exists an $\epsilon>0$ such that $\phi(\lambda)$ is thermodynamically unstable for all 
 $\lambda \in (0,\epsilon)$. Similarly, if $\sigma < 0$, then there exists an $\epsilon>0$ such that $\phi(\lambda)$ is thermodynamically unstable for all
  $\lambda \in (-\epsilon, 0)$.
\end{theorem}

\begin{proof}

Since $\lambda=0$ is a turning point, the quantities $dS/d\lambda$ and $dX^i /d \lambda$ vanish at $\lambda = 0$. Hence, $(dS/d\lambda)/\lambda$ and $(dX^i /d \lambda)/\lambda$ are smooth functions of $\lambda$. By Condition \ref{condA}, we can choose a 1-parameter family of perturbations, $\widetilde\delta\phi(\lambda)$, (assumed to be smooth in $\lambda$) such that in a neighborhood of $\lambda=0$, $\widetilde\delta\phi(\lambda)$ is a linearized solution off of $\phi(\lambda)$ such that
\be \label{tildedef}
\ba
\widetilde \delta S(\lambda) &= \frac{1}{\lambda}\frac{dS}{d\lambda}\\
\widetilde \delta X^i (\lambda) &= \frac{1}{\lambda}\frac{dX^i}{d\lambda} \, ,
\ea
\ee
where we have used the first law together with $\mathcal T (\lambda=0) \neq 0$ to effectively interchange $M$ and $S$ in the formulation of Condition \ref{condA}. Let
\be\label{hatdef}
\widehat\delta\phi(\lambda) = \frac{d\phi}{d\lambda} - \lambda \widetilde\delta\phi(\lambda) \, .
\ee
Then, by construction, the perturbation $\widehat\delta\phi(\lambda)$ has $\widehat\delta S (\lambda) = \widehat\delta X^i (\lambda) = 0$ (and, hence, by the first law, $\widehat \delta M(\lambda)=0$). Using the symmetry and bilinearity of $\mathcal B$, we have
\be
\ba
\mathcal B&\left[\phi(\lambda); \widehat\delta\phi(\lambda), \widehat\delta\phi(\lambda)\right] \\
&= \mathcal B\left[\phi(\lambda); \frac{d\phi}{d\lambda}, \frac{d\phi}{d\lambda}\right] - 2 \lambda \mathcal B\left[\phi(\lambda); \frac{d\phi}{d\lambda}, \widetilde\delta\phi(\lambda)\right] + \lambda^2 \mathcal B\left[\phi(\lambda);\widetilde\delta\phi(\lambda), \widetilde\delta\phi(\lambda)\right] \, .
\ea
\ee
Using the lemma on the first two terms, we obtain
\be
\ba
\mathcal B&\left[\phi(\lambda); \widehat\delta\phi(\lambda), \widehat\delta\phi(\lambda)\right] \\
&= \left(\frac{d\mathcal T}{d\lambda} \frac{dS}{d\lambda} + \frac{d Y_i}{d\lambda} \frac{dX^i}{d\lambda}\right) - 2 \lambda \left( \frac{d\mathcal T}{d\lambda} \widetilde\delta S + \frac{d Y_i}{d\lambda} \widetilde\delta X^i \right) + \lambda^2 \mathcal B\left[\phi(\lambda);\widetilde\delta\phi(\lambda), \widetilde\delta\phi(\lambda)\right] \\
&=-\left(\frac{d\mathcal T}{d\lambda} \frac{dS}{d\lambda} + \frac{d Y_i}{d\lambda} \frac{dX^i}{d\lambda}\right)+\lambda^2 \mathcal B\left[\phi(\lambda);\widetilde\delta\phi(\lambda), \widetilde\delta\phi(\lambda)\right] \, .
\ea
\ee
Taylor expanding $dS/d\lambda$ and $dX^i /d \lambda$ about $\lambda = 0$, we obtain
\be\label{Bwithhats}
\mathcal B\left[\phi(\lambda); \widehat\delta\phi(\lambda), \widehat\delta\phi(\lambda)\right] = 
- \lambda \sigma + O(\lambda^2) \, ,
\ee 
from which it can be immediately seen that if $\sigma > 0$, then $\mathcal B < 0$ for $\lambda \in (0,\epsilon)$ for some $\epsilon > 0$.
\end{proof}

It should be noted that this theorem can be generalized along the lines of Theorem 2 in \cite{Sorkin2}. Namely, if $\phi(\lambda, \lambda')$ is a 2-parameter family of thermodynamic equilibrium solutions such that $(\lambda,\lambda')=(0,0)$ is a turning point of the curve $\lambda'=0$, and if 
\be
\left.\frac{\partial}{\partial \lambda'} \left( \frac{\partial{\mathcal T}}{\partial \lambda} \frac{\partial S}{\partial \lambda}+
 \frac{\partial{Y_i}}{\partial \lambda} \frac{\partial X^i}{\partial \lambda} \right)\right|_{\lambda=\lambda'=0}>0,
\ee
then $\phi(0,\lambda')$ is thermodynamically unstable for all $\lambda' \in (0,\epsilon)$. The proof of this generalization is essentially the same as the above proof. The main difference is that one now defines
\be
\widehat\delta \phi(\lambda, \lambda') = \frac{\partial \phi}{\partial \lambda} - \lambda \widetilde \delta \phi (\lambda,\lambda') - \lambda' \widetilde \delta' \phi (\lambda,\lambda'), 
\ee
where  $\widetilde \delta \phi$ and $\widetilde \delta' \phi$ are chosen to make $\widehat\delta S (\lambda,\lambda') = \widehat\delta X^i (\lambda,\lambda') = 0$.

\section{Turning Points and Dynamic Instability}\label{rings}

As is clear from the previous section, the existence of a turning point yields a sufficient condition for the existence of a thermodynamic instability along a family of thermodynamic equilibrium solutions. However, turning points do not provide a necessary condition for thermodynamic instability, i.e., the onset of a thermodynamic instability could occur prior to a turning point or without the presence of any turning point at all. The situation is considerably worse for using turning points to determine the presence of dynamic instabilities\footnote{\footnotesize See \cite{GSW} for a complete discussion of the relationship between dynamic and thermodynamic instabilities. \medskip}: Again, the absence of a turning point does not imply the absence of a dynamic instability. Indeed, numerical investigations by Takami, Rezolla, and Yoshida \cite{Takami} found that dynamic instability---and therefore also thermodynamic instability---sets in \emph{prior} to the turning point along families of isentropic, rigidly rotating, perfect fluid stars. Furthermore, since thermodynamic instability does not, in general, imply dynamic instability, the presence of a turning point does not, in general, imply the existence of a dynamic instability. 

However, there are some situations of interest where thermodynamic instability implies dynamic instability, and in those situations the presence of a turning point can be used to infer the existence of a dynamic instability. The first is the case of spherically symmetric, relativistic stars with an ``isentropic'' equation of state\footnote{\footnotesize This case encompasses a general ``barotropic'' equation of state $P = P(\rho)$. \medskip}, $\rho = \rho(n)$. It was shown in \cite{GSW} that, in this case, if $\phi$ is a static, spherically symmetric solution and $\mathcal B < 0$ for a spherically symmetric perturbation for which $\delta N = 0$, then a {\it dynamic} instability must be present (see also \cite{Roupas}). Consequently, the presence of a turning point along a sequence of static, spherically symmetric solutions with isentropic equation of state implies the existence of a dynamic instability.

A second case is that of black holes in vacuum general relativity in arbitrary spacetime dimension $D \geq 4$. It was proven in \cite{HollandsWald} that, in this case, positivity of $\mathcal B$ for axisymmetric perturbations with $\delta M = \delta J^i = 0$ (where $J^i$ denote the angular momenta of the independent planes of rotation) is the criterion for dynamic as well as thermodynamic stability. Thus---assuming that stationary black hole solutions satisfy Condition \ref{condA}---the presence of a turning point implies the presence of a dynamic instability. Figueras, Murata, and Reall \cite{Reall}
have analyzed turning points for black rings in $5$ spacetime dimensions. In the remainder of this section, we relate their results to ours\footnote{\footnotesize It is possible to similarly relate the stability results of \cite{Reall2} for non-uniform black string solutions in 12 and 13 spacetime dimensions to our turning point results. \medskip}.

Figueras et al.~\cite{Reall} numerically analyze a 3-parameter family of ``doubly spinning'' black rings in 5 dimensions. The thermodynamic state parameters are the mass, $M$ and the independent angular momenta, $J^1$ and $J^2$. They find that there is a surface of turning points in the sense of \eqref{jac} separating a region of negative $c_J$ from a region of positive $c_J$, where  
\be
c_J \equiv \left( \frac{\partial M}{\partial \mathcal T} \right)_J =\mathcal T \left( \frac{\partial S}{\partial \mathcal T} \right)_J
= \frac{\kappa}{4} \left( \frac{\partial A}{\partial \kappa} \right)_J
\ee
is the heat capacity at constant $J$. They find numerically that instability occurs on the side where $c_J > 0$. This can be understood as follows.

We have verified that the quantities $(\mathcal T, J^1, J^2)$ smoothly and non-degenerately parameterize the family of black rings\footnote{\footnotesize The parametrization $(k, \alpha, \nu)$ used by \cite{Reall} can be verified to be non-degenerate. Using the formulas for $(\mathcal T, J^1, J^2)$ given in \cite{Elvang}, we have confirmed that the Jacobian determinant $\left|\partial(\mathcal T, J^1, J^2)/\partial(k,\alpha,\nu)\right|$ is nowhere vanishing.\medskip}.
Thus turning points in the sense of \eqref{jac} occur precisely at the points at which
\be
0 = \left|\frac{\partial(\mathcal M, J^1, J^2)} {\partial(\mathcal T, J^1, J^2)}\right| = c_J \, .
\ee
Furthermore, it is easily seen that any point at which $c_J = 0$ is a turning point of a curve of fixed angular momenta. 
Since for this curve, we have
\be
\sigma = \frac{d}{d\lambda}\left(\frac{d{\mathcal T}}{d \lambda} \frac{d S}{d \lambda}\right) \bigg|_{\lambda=0} \, ,
\ee
by the turning point theorem, a thermodynamic instability---and, hence, a dynamic instability \cite{HollandsWald}---will be present on the side of the turning point where
\be
\frac{d\mathcal T}{d\lambda} \frac{d S}{d\lambda} > 0  \, .
\ee
But we have
\be
\frac{d S}{d\lambda} =  \left(\frac{\partial S}{\partial \mathcal T}\right)_J \frac{d \mathcal T}{d\lambda} =  \frac{c_J}{\mathcal T} \frac{d \mathcal T}{d\lambda} \, ,
\ee
from which it follows that instability will occur on the side where $c_J$ is positive, in agreement with \cite{Reall}. Of course, while \cite{Reall} found numerically that all of the solutions with $c_J>0$ are unstable, here we have only shown instability in a small neighborhood of $c_J=0$.

\begin{acknowledgments}
This research was supported in part by NSF grant PHY~12-02718 to the University of Chicago.
\end{acknowledgments}
 \newpage
\bibliography{mybib}

\end{document}